%% file: main.tex
\documentclass[conference,letterpaper]{IEEEtran}

\IEEEoverridecommandlockouts

\usepackage{cite}
\usepackage{amsmath,amssymb,amsfonts}
\usepackage{graphicx}
\usepackage{textcomp}

\usepackage{xcolor}
\def\BibTeX{{\rm B\kern-.05em{\sc i\kern-.025em b}\kern-.08em
    T\kern-.1667em\lower.7ex\hbox{E}\kern-.125emX}}

\input{Macros}

\begin{document}

\title{On the Error Probability of RPA Decoding of Reed-Muller Codes over BMS Channels 
}

\author{\IEEEauthorblockN{Dorsa Fathollahi}
\IEEEauthorblockA{{Department of Electrical Engineering} \\
{Stanford University, CA, USA}\\
email: \texttt{dorsafth@stanford.edu}
}
\and
\IEEEauthorblockN{V. Arvind Rameshwar}
\IEEEauthorblockA{{Department of Electrical Engineering} \\
{IIT Madras, India}\\
email: \texttt{arvind@ee.iitm.ac.in}}
\and
\IEEEauthorblockN{V. Lalitha}
\IEEEauthorblockA{{SPCRC} \\
{IIIT Hyderabad, India}\\
email: \texttt{lalitha.v@iiit.ac.in}}
}

\maketitle

\begin{abstract}
We analyze the performance of the Recursive Projection-Aggregation (RPA) decoder of Ye and Abbe (2020), for Reed-Muller (RM) codes, over general binary memoryless symmetric (BMS) channels. Our work is a significant generalization of a recent result of Rameshwar and Lalitha (2025) that showed that the RPA decoder provably achieves vanishing error probabilities for ``low-rate" RM codes, over the binary symmetric channel (BSC). While a straightforward generalization of the proof strategy in that paper will require additional, restrictive assumptions on the BMS channel, our technique, which employs an equivalence between the RPA projection operation and a part of the ``channel combining" phase in polar codes, requires no such assumptions. Interestingly, such an equivalence allows for the use of a generic union bound on the error probability of the first-order RM code (the ``base case" of the RPA decoder), under maximum-likelihood decoding, which holds for any BMS channel. We then exploit these observations in the proof strategy outlined in the work of Rameshwar and Lalitha (2025), and argue that, much like in the case of the BSC, one can obtain vanishing error probabilities, in the large $n$ limit (where $n$ is the blocklength), for RM orders that scale roughly as $\log \log n$, for all BMS channels.

\end{abstract}


\section{Introduction}
Reed-Muller (RM) codes are a well-studied family of binary linear codes that are obtained by the evaluations of Boolean polynomials on the points of the Boolean hypercube \cite{reed,muller}. Recent breakthrough theoretical progress has shown that RM codes are  capacity-achieving for the binary erasure channel \cite{kud1}, and more generally, for BMS channels \cite{Reeves,abbesandon} (we refer the reader to the survey \cite{rm-survey} for a detailed treatment of the properties of RM codes). 


Over the past few decades, much work has been dedicated to finding practical (low-complexity) decoding algorithms for RM codes. The earliest such algorithm by Reed \cite{reed} is capable of correcting bit-flip errors up to half the minimum distance of the code. For the case of first-order RM codes, a Fast Hadamard Transform-based (or FHT-based) decoder was designed in \cite{fht2,fht}, which is an efficient implementation of a maximum likelihood (ML) decoding procedure. We also refer the reader to decoding algorithms for RM codes of higher code that display good performance at moderate blocklengths in the works \cite{sidelnikov,sakkour,dumer1,dumer2,dumer3,dumer_burnashev}. 

A much more recent decoding algorithm, variants of which were shown to achieve near-ML performance at moderate blocklengths over the binary symmetric channel (BSC), is the Recursive Projection-Aggregation (RPA) decoder of Ye and Abbe \cite{rpa}. 
Later works \cite{rpacomplex1,rpacomplex2} presented procedures for reducing the complexity of the RPA decoder, by making use of a subset of the subspaces employed by the RPA decoder, for projection. More recent work \cite{coset-error-it}, however, argues that limiting the number of subspaces used could suffer from significant performance loss, due to the presence of ``coset error patterns". In this work, we hence work with the original RPA decoder of Ye and Abbe, which makes use of all subspaces of a fixed dimension, for projection.

Our main efforts in this paper are directed towards obtaining \emph{analytical} performance guarantees for RPA decoding, via explicit bounds on the probability of error, over general binary memoryless symmetric (BMS) channels. Recent work \cite{rpa-analysis-isit} has obtained theoretical upper bounds on this error probability for the case when the channel is a binary symmetric channel (BSC). In this work, we generalize the results in \cite{rpa-analysis-isit} to the broad class of BMS channels, of which the BSC is a part. We mention however that such an extension does not follow straightforwardly from the proof strategy of \cite{rpa-analysis-isit} -- indeed, while somewhat direct modifications can be carried out by placing additional restrictions on the channel (such as requiring a bounded output alphabet size and bounded log-likelihood ratios, for all outputs) -- the problem of deriving general bounds for \emph{arbitrary} BMS channels requires different techniques. 

Our main result (Theorem \ref{thm:main}) shows that RPA decoding for general BMS channels guarantees vanishing error for RM orders roughly logarithmically in the parameter $m = \log n$, where $n$ is the code blocklength -- a scaling that is asymptotically identical to that obtained in \cite{rpa-analysis-isit}. Our proof proceeds via the identification of an equivalence between the projection operation in RPA decoding and a certain channel combining operation in polar code construction, which allows for the use of a standard union bound on the ML error probability of first-order RM codes, which form the ``base case" of the RPA decoder. Crucially, our analysis retains an important element of the proof strategy of \cite{rpa-analysis-isit}, which restricts attention to the case when one iteration of the RPA decoder suffices for convergence. It hence appears that relaxing this restriction and performing an analysis of RPA decoding in the presence of the correlations introduced via multiple iterations is key to obtaining asymptotic improvements in error probability.



\section{Notation and Preliminaries}
\subsection{Notation}
\label{sec:notation}
Random variables are denoted by capital letters, e.g., $X, Y$, and small letters, e.g., $x, y$, denote their instantiations. Log-likelihood vectors are however denoted as $L$, following standard notation.
The notation $\mathbf{0}$ denotes the all-zeros vector, whose length can be inferred from the context.
Natural logarithms are denoted as $\ln$. The notations $O(\cdot), o(\cdot), \Omega(\cdot),\omega(\cdot)$ are used to refer to members of the standard  Bachmann–Landau family of asymptotic notations. The indicator function $\indic{\cdot}$ takes the value $1$ when the argument is true, and $0$, otherwise.

\subsection{Reed-Muller Codes}
Consider the polynomial ring $\mathbb{F}_2[x_1,x_2,\ldots,x_m]$ in $m$ variables. For a polynomial $f\in \mathbb{F}_2[x_1,x_2,\ldots,x_m]$ and a binary vector $\mathbf{z} = (z_1,\ldots,z_m)\in \mathbb{F}_2^m$, we write {$f(\mathbf{z})=f(z_1,\ldots,z_m)$} as the evaluation of $f$ at $\mathbf{z}$. Let $\mathbb{F}_2^{\leq r}[x_1,x_2,\ldots,x_m]$ denote the collection of polynomials of degree at most $r$. The evaluation points are ordered according to the standard lexicographic order on strings in $\mathbb{F}_2^m$, i.e., if $\mathbf{z} = (z_1,\ldots,z_m)$ and $\mathbf{z}^{\prime} = (z_1^{\prime},\ldots,z_m^{\prime})$ are two evaluation points, then, $\mathbf{z}$ occurs before $\mathbf{z}^{\prime}$ iff for some $i\geq 1$, we have $z_j = z_j^{\prime}$ for all $j<i$, and $z_i < z_i^{\prime}$. Now, let Eval$(f):=\left({f(\mathbf{z})}:\mathbf{z}\in \mathbb{F}_2^m\right)$ be the evaluation vector of $f$, where the coordinates $\mathbf{z}$ are ordered according to the standard lexicographic order. 

\begin{definition}[see Ch. 13 in \cite{mws}, or \cite{rm_survey}]
	{For $0\leq r\leq m$}, the $r^{\text{th}}$-order binary Reed-Muller code RM$(m,r)$ is defined as
	\[
	\text{\normalfont RM}(m,r):=\{\text{\normalfont Eval}(f): f\in \mathbb{F}_2^{\leq r}[x_1,x_2,\ldots,x_m]\}.
	\]
\end{definition}
All through, we set $n:=2^m$.
 \subsection{BMS Channels and RPA Decoding}
 We refer the reader to \cite[Ch. 4]{mct} for the definition of binary memoryless symmetric (BMS) channels. 
 \begin{definition}
 The Bhattacharyya parameter $Z(W)$ of a BMS channel $W$ with output alphabet $\mathcal{Y}$ is defined as 
    \[
    Z(W):=  \sum_{y\in \mathcal{Y}} \sqrt{P_{Y|0}(z|0)P_{Y|1}(z|1)},
    \]
    when $\mathcal{Y}$ is finite, with the conditional p.m.f.s above replaced by corresponding probability densities and the summation above replaced by an integral, when the channel laws correspond to continuous probability distributions.
\end{definition}
 The RPA decoding algorithm for general BMS channels is briefly recapitulated as Algorithm \ref{alg:RPA-BMS-r1} (see also \cite[Alg. 3]{rpa})\footnote{Algorithm 3 of \cite{rpa} in fact is stated for general binary-input memoryless channels that are not necessarily symmetric.}. In all that follows in this paper, we assume that the RPA decoder uses one-dimensional subspaces for projection.
 
 In Algorithm \ref{alg:RPA-BMS-r1}, $L$ denotes the vector of log-likelihood ratios (LLRs) with $L = (L(z):\ z\in \{0,1\}^m)$, where $L(z):=\log \frac{W(y_z|0)}{W(y_z|1)}$, and $\mathbf{y} = (y_z:\  z\in \{0,1\}^m)$ is the received vector at the end of the BMS channel. In Step 2, we use
 \begin{align*}
     L^{(v)}(T)&:=\log\big(\exp(L(z){+}L(z\oplus v))+1\big)\\
                &\ \ \ \ -\log\big(\exp(L(z))+\exp(L(z\oplus v))\big).
  \end{align*}
  We denote $L^{(v)} := (L^{(v)}(T): T \in \{0,1\}^m / \langle v\rangle)$ for some fixed ordering among cosets $T$. Similarly $\widehat{y}^{(v)} := \widehat{y}^{(v)}(T) : T \in \{0,1\}^m / \langle v\rangle)$.
  In Step 4, the subroutine $\mathrm{FHTDecoder}$ refers to the standard Fast Hadamard Transform decoder \cite{fht2,fht} for ML decoding of first-order RM codes.
 Furthermore, in this paper, we do not set a ``tolerance threshold" $\theta$ as in the original work \cite{rpa}, and instead only run the decoder for a fixed number $N_\text{max}$ of iterations. However, for the purpose of analysis, as in \cite{rpa-analysis-isit}, we restrict attention to the setting where \emph{one} iteration of RPA decoding suffices for returning the correct decoded estimate.

 Given the code RM$(m,r)$, let $C$ denote a (random) codeword that is drawn uniformly at random from the code. Now, let $\mathbf{Y}$ be the received sequence at the end of the BMS channel $W$ and let $\widehat{C}$ denote the estimate of the input codeword obtained by using the RPA decoder in \cite[Alg. 3]{rpa}. The probability of error of RM$(m,r)$ under RPA decoding is then defined as $P_\text{err}(\text{RM}(m,r)):=\Pr\left[\widehat{C}\neq C\right]$.

 \begin{definition}[RPA Recursion Tree]
The RPA decoding of $\text{\normalfont RM}(m, r)$ induces a recursion tree $\mathcal{T}$ with the following structure:
\begin{itemize}
    \item The tree has depth $r$, with orders indexed by $i \in \{1, 2, \ldots, r\}$.
    \item Nodes are uniquely identified by a pair $v = (i, j)$ where:
    \begin{itemize}
        \item $i \in \{1, 2, \ldots, r\}$ is the \emph{height} of the node, where the root is at height $r$
        \item $j \in \{1, 2, \ldots, N_i\}$ is the \emph{index} at height $i$; there exists a one-one correspondence between $j$ and the subspace $v$ used for projection
    \end{itemize}
    \item A node $v = (i, j)$ at order $i$ corresponds to the RM code $\text{\normalfont RM}(m - r + i, i)$.
\end{itemize}
\end{definition}
Owing to the close relationship between the recursions in a single iteration of the RPA decoder and the tree structure defined above, we let $L_i^{(j)}$ denote the LLR vector $L^{(v)}$ at the $(r-i)^\text{th}$ step of the recursion, using the subspace $v$ associated with index $j$ of the node $(i,j)$, for $i\in \{1,\ldots,r\}$. The estimates $\widehat{L}_i^{(j)}$ and $\widehat{y}_i^{(j)}$ are similarly defined.

The next lemma considers the channel
\begin{align*}
 & W^-(y_1,y_2 \mid s)
 \\ & \;=\;
  \frac{1}{2}
  \sum_{\substack{u_1,u_2 \in \bits \\ u_1 \oplus u_2 = s}}
    W(y_1 \mid u_1)\,W(y_2 \mid u_2),
  \quad s \in \bits,
\end{align*}
which, as we argue, is precisely the channel induced via the projection operation.
 \begin{algorithm}[t]
\caption{RPA Decoder for RM$(m,r)$}
\label{alg:RPA-BMS-r1}
\textbf{Input:} LLR vector $L\in\mathbb{R}^{n}$, max iter. $N_{\max}$\\
\textbf{Output:} $\widehat{c}\in\{0,1\}^{n}$.

\begin{algorithmic}[1]
\For{$j=1$ \textbf{to} $N_{\max}$}
  \State \textbf{Projection:} For each $v\in \mathbb{F}_2^m\!\setminus\!\{0\}$ and each coset $T=\{z,z\oplus v\}$, compute $L^{(v)}(T)$.
  
  \If{$r\geq 2$} Compute $\widehat{y}^{(v)} \leftarrow \mathrm{RPA}(m{-}1,r{-}1,L^{(v)})$.
  \Else\ Set $\widehat{c}\leftarrow \mathrm{FHTDecoder}(L)$ \textbf{break}
\EndIf

  \State \textbf{Aggregation:} For each $z\in\mathbb{F}_2^m$, compute
  \[
     \widehat{L}(z) \gets \frac{1}{2^m-1}\sum_{v\neq 0}\Bigl(1-2\,\widehat{y}^{(v)}([z+\langle v\rangle])\Bigr)\,L(z\oplus v).
  \]
  \State Set $L\gets \widehat{L}$.
\EndFor
\State Set $\widehat{c}(z)\gets \indic{L(z)<0}$ for all $z$; \Return $\widehat{c}$.
\end{algorithmic}
\end{algorithm}
\begin{lemma}\label{lem:induced-channel}
Let \(W \colon \bits \to \cY\) be a binary-input memoryless symmetric (BMS) channel.  
Then:
\begin{enumerate}
    \item $W^-$ is a BMS channel. 
    \item 
  For any nonzero \(v \in \F_2^m\), the pair \(\bigl(Y(z),Y(z\oplus b)\bigr)\) conditioned on \(C(z) \oplus C(z\oplus b)\) has distribution \(W^-\).
  \item The Bhattacharyya parameter satisfies \(Z(W^-)\,\le\,1 - \bigl(1 - Z(W)\bigr)^2\).
\end{enumerate}
\end{lemma}
\begin{proof}
  (1) Follows directly from \cite[Prop. 13]{polar}.\\
(2)  Let \(U_1,U_2\sim\mathrm{Bern}(1/2)\) be independent and set \(S=U_1\oplus U_2\).  It can easily be checked that $\PR{Y_1=y_1,Y_2=y_2 \mid S=s} = W^-(y_1,y_2 \mid s),$ for $s\in \{0,1\}$.
Any projection \(v\neq \mathbf{0}\) selects a coset \(\{z,z\oplus v\}\), and because the channel uses are memoryless, the distribution of \(\bigl(Y(z),Y(z\oplus b)\bigr)\) given the parity \(\cC(z)\oplus \cC(z\oplus b)\) is \(W^-\).  Thus the induced channel for projection along \(b\) is \(W^-\).
\\
(3)  Follows directly from \cite[Prop. 5]{polar}.
\end{proof}
Following Lemma \ref{lem:induced-channel}, we see that the channel induced by projecting along any one-dimensional subspace \(\{0,v\}\), which we denote $W^{(v)}$, is exactly \(W^-\). Following previous notation, we let $W_i$ stand for the (common) channel induced after any projection at the $(r-i+1)^\text{th}$ step of the RPA recursion, $i\in \{1,\ldots,r\}$. Let $ Z_i$ be the Bhattacharyya parameter of channel $W_i$.
\begin{lemma}\label{lem:bhat-rec} 
    We have that $Z_i \leq 1- ( 1- Z_{i+1})^2.$
        
\end{lemma}
\begin{proof}

Fix any node at height \( i+1 \), and consider the projection along a one-dimensional subspace \(B = \{0 ,v \} \).  By Lemma~\ref{lem:induced-channel} (Items 2 and 3), we obtain that 
$
Z(W_{i+1}^-)\;\le\;1-\bigl(1-Z(W_{i+1})\bigr)^2.
$
Since \(Z_i=Z(W_i)=Z(W^{(v)})\), we obtain the desired bound.
\end{proof}
\begin{lemma}\label{lem:bhat-unroll}
We have that for any $i\in \{1,\ldots,r\}$, $ Z_i  \leq 1- (1- Z_r)^{2^{r-i}}$.
\end{lemma}
\begin{proof}
    From Lemma~\ref{lem:bhat-rec}, we know that $ Z_i \leq 1-(1-Z_{i+1})^2 $.  Define $Y_i = 1 - Z_i$. Then, $ Y_i = (1-Z_{i+1})^2 = Y_{i+1}^2$. Therefore, $Z_i = 1-Y_i = 1-Y_r^{2^{r-i}} = 1-(1-Z_r)^{2^{r-i}}$.
\end{proof}

\section{Main Result}

As in \cite{rpa-analysis-isit}, we use \cite[Prop. 2]{rpa} to focus our analysis on the case when the input codeword is fixed to be $C = \mathbf{0}\in \{0,1\}^N$, since, via the symmetry of the channel, we have that $P_\text{err}(\text{RM}(m,r))$ equals the error probability of the all-zeros codeword. In what follows, we let $Z=:Z(W)\in (0,1)$ denote the Bhattacharyya parameter of the BMS channel $W$ of interest\footnote{We restrict our attention in this paper to ``non-degenerate" BMS channels whose Bhattacharyya parameters are strictly bounded away from $0$ and $1$.}.
\begin{theorem}\label{thm:main}
Let
$
\lambda := -\ln(1-Z)\in(0,\infty).
$
For
$
r < \log_2 m -\log_2 \lambda,
$
we have $P_\text{err}(\text{RM}(m,r))\xrightarrow{m\to \infty} 0$.
\end{theorem}
\begin{remark}
    Consider the setting where the BMS channel $W$ is the BSC$(p)$, with the cross-over probability $p\in (0,1/2)$. Theorem III.1 of \cite{rpa-analysis-isit} shows that for $r<\ln m + \ln\left(\frac{\ln 2}{\ln (1-2p)}\right)$, the RPA decoder achieves vanishing error probabilities over the BSC$(p)$. Using the fact that for this channel, we have $Z = 2\sqrt{p(1-p)}$, it can be checked via numerical comparisons that the claim in \cite{rpa-analysis-isit} is stronger than Theorem \ref{thm:main}, for all $p\in (0,1/2)$. However, both claims provide identical asymptotic guarantees on the growth rate of $r$ with $m$ (i.e., $r$ growing roughly logarithmically in $m$) for vanishing error probabilities under RPA decoding.
\end{remark}
\section{Helper Lemmas}
In order to prove Theorem \ref{thm:main}, we shall first establish an upper bound on the error probability at each stage of recursion in terms of the error probabilities of the previous stages.  We then unroll this recursion to obtain the final block error probability. But first, we require some more notation. For $1\leq i\leq r$, let $N_i:= 2^{m-r+i}$. Thus, the number of subspaces used for projection at any node at height $i$ is $N_i-1$.

\begin{definition}
     For any $1\leq i\leq r$ and $j\leq N_i$, let $ \cQ_i^{(j)}$ be the event that the decoded estimate $\widehat{Y}_i^{(j)}$ at node $(i,j)$ of $\mathcal{T}$ is incorrect (does not equal $\mathbf{0}$). For ease of reading, we write $ \cQ_i^{(j)}$ as $\cQ_i$, when the index $j$ is clear from the context. 
\end{definition}
We are interested in obtaining an upper bound on $\PR{ \cQ_r}$, which directly yields an upper bound on $P_\text{err}(\text{RM}(m,r))$.
 
\begin{definition}

For a node $(i, j)$ at height $i \geq 2$, define the event $\cG_{i}^{(j)} $ as the event that all children of node $ (i,j) $ with order $i+1$ are decoded correctly.  
\end{definition}
\subsection{Recurrence Relations for $\PR{ \cQ_i}$}
We proceed with deriving a recurrence relation for $\PR{ \cQ_i}$, $1\leq i\leq r$, in terms of $\PR{ \cQ_j}$, $j<i$; in all that follows, we implicitly condition on the fact that the all-zeros codeword was transmitted. We then ``unroll" this recurrence to yield a closed-form upper bound on $\PR{ \cQ_i}$. For any event $\mathcal{E}$, we let $\overline{\mathcal{E}}$ denote its complement, where the universe can be inferred from the context.

First, we obtain an upper bound on $\PR{ \cQ_1}$, which forms the ``base case" of our recursive analysis of error probabilities. While the upper bound can also be obtained via the application of a standard union bound argument for the ML error probability (see, e.g., \cite[Sec. 2.1]{sason-shamai} or \cite[Problem 1.21]{mct}), we provide a direct proof, using properties of first-order RM codes, which could be of independent interest.
\begin{lemma}\label{lem:Q1}
     We have that
$
        \PR{ \cQ_1}  \leq (2 ^ { m -r + 2  } - 1 ) Z_1^{2^{m-r}}.
$

\end{lemma}
\begin{proof}

Via standard arguments (see, e.g., \cite[Sec. IV-A]{rpa-analysis-isit}), there exists a one-one correspondence between the codewords of $\text{\normalfont RM}(m-r+1,1)$ with the functions 
$\sigma\cdot \chi_s$ where $s \in \{0,1\}^{m-r+1}$ and $\sigma \in \{\pm 1\}$, with $\chi_{{s}}({x}):=(-1)^{{x}\cdot {s}}$, ${x}\in \{0,1\}^{m-r+1}$.  
The Fast Hadamard Transform (FHT) used by the RPA decoder then computes
\begin{equation*}
    (\sigma, s) =\argmax_{\sigma, s} \langle L , \sigma\cdot \chi_s \rangle.
\end{equation*}
Here, for functions $f, g: \{0,1\}^n\to \{-1,1\}$, we define their inner product
$
\langle f,g\rangle:= \frac{1}{2^n}\cdot \sum_{{x}\in \{0,1\}^n} f({x})g({x}),
$
and $L$ denotes the vector of log-likelihood ratios $(L(z))$, for $z\in \{0,1\}^{m-r+1}$, with $L(z):=\log \frac{W(Y_z|0)}{W(Y_z|1)}$.
It can be checked that the ``correct" codeword for any leaf node of $\mathcal{T}$ is the all-zeros codeword $\mathbf{0}$, which corresponds to $(\sigma,s)= (1,\mathbf{0})$. For any pair $(\sigma,s)\neq ( 1 , \mathbf{0})$, define the event, $ \cE_{\sigma,s} = \indic{ \langle L,\sigma \cdot \chi_s \rangle\geq \langle L, \chi_{\mathbf{0}} \rangle  }$. The FHT decoder makes an error iff $ \cE_{\sigma,s}$ occurs, for some $(\sigma,s)\neq ( 1 , \mathbf{0})$. Thus, via a union bound, 
\begin{equation}
     \PR{ \cQ_1}  \leq   \sum_{(\sigma, s)\neq (1,\mathbf{0})}  \PR{ \cE_{\sigma,s} }. \label{eq:3}
\end{equation}
Observe that
$
\langle L , \sigma\cdot \chi_s \rangle - \langle L , \chi_{\mathbf{0}} \rangle
    = - \sum_{z : \sigma\cdot \chi_s(z) = -1} L(z).
$
Hence the event $ \cE_{\sigma ,s }$ is equivalent to, 
\begin{align}
\label{eq:1}
    \cE_{\sigma,s} & = \indic{ \sum_{z : \sigma\cdot \chi_s(z) = -1} L(z) < 0 } \nonumber \\
    & = \indic{ e^{- \frac{\sum_{z : \sigma\cdot \chi_s(z) = -1} L(z)}{2}} >  1}.
\end{align}
Therefore by applying the Markov inequality to \eqref{eq:1}, we have
\begin{align}
\PR{\cE_{\sigma,s}} &\leq \E{ e^{- \frac{\sum_{z : \sigma\cdot \chi_s(z) = -1} L(z)}{2}}} \notag\\ &= \prod_{z : \sigma\cdot \chi_s(z) = -1} \E{ e^{- \frac{ L(z)}{2}}}. \label{eq:2}
\end{align}
We reiterate that all probabilities and expected values above are conditioned on $\mathbf{0}$ being transmitted. The last step holds since $L(z)$ are i.i.d. across $z\in \{0,1\}^{m-r+1}$.
Now, observe that 
\begin{align*}
\E{e^{-L(z)/2}}
    &= \sum_{y} W(y|0) \sqrt{\frac{W(y|1)}{W(y|0)}} \\
    &= \sum_{y} \sqrt{W(y|0) W(y|1)} = Z_1.
\end{align*}
Therefore, following on from \eqref{eq:2},
\begin{equation*}
    \prod_{z : \sigma \chi_s(z) = -1} \E{ e^{- \frac{L(z)}{2}}} = Z_1^{\lvert\{z : \sigma\cdot \chi_s(z)=-1\}\rvert }.
\end{equation*}
We now bound the exponent ${\lvert\{z : \sigma\cdot \chi_s(z)=-1\}\rvert}$. 
For any nonzero $s$, $\chi_s$ has equal number of $+1 $ and $-1$. As a result, ${\lvert\{z : \sigma\cdot \chi_s(z)=-1\}\rvert}  = 2^{m-r}$ for all $ s\neq \mathbf{0}$ because multiplying by $\sigma$ flips the signs globally. Moreover, when $ s = \mathbf{0} $, we have that when $ \sigma = -1 $, ${\lvert\{z : \sigma \cdot \chi_s(z)=-1\}\rvert} = 2^{m-r+1}$. Therefore, 
\begin{equation*}
    \PR{\cE_{\sigma,s}}\le Z_1^{2^{m-r}}, \ \text{for all } (\sigma,s) \neq (1,0).
\end{equation*}
Finally, via \eqref{eq:3}, we obtain that
$
        \PR{\cQ_1}  \leq (2 ^ { m -r + 2  } - 1 ) Z_1^{2^{m-r}}.
$
\end{proof}

\begin{lemma}
 \label{lem:Qi-rec}
For $ i \geq  2 $, we have
    \begin{equation*}
        \PR{\cQ_i } \leq N_i Z_i^{N_i -1 }  + (N_i - 1 ) \PR{\cQ_{i-1}}.
    \end{equation*}
\end{lemma}
\begin{proof}
We know that
\begin{equation}
    \PR{\cQ_i } \leq \PR{\cQ_i | \cG_i } + \PR{\overline{\cG_i}}. 
\end{equation}
We will prove the lemma by showing the following two claims hold: Firstly, that
\begin{equation}
\label{eq:4}
    \PR{\overline{ \cG_i}} \leq (N_i -1 ) \PR{\cQ_{i-1}},
\end{equation}
and secondly, that 
\begin{equation}
\label{eq:5}
    \PR{\cQ_i | \cG_i } \leq N_i Z_i ^{ N_i -1 }.
\end{equation} 
To prove the first claim, notice that the event $\cG_i$ is equivalent to the event $ \cap_{j' \in \textbf{ch}( i,j ) }  Q_{i+1}^{j'}$, where $\textbf{ch}( i,j )$ denotes the collection of children of node $(i,j)$. Hence by the union bound, we get that
\begin{equation*}
    \PR{\overline{ \cG_i}} \leq \sum_{j'\textbf{ch}( i,j ) } \PR{\cQ_{i+1}^{(j)}} \leq ( N_i-1 ) \PR{\cQ_{i+1}},
 \end{equation*}
thereby proving \eqref{eq:4}.

Now we move on to proving the second claim \eqref{eq:5}. 
 Let $\mathcal{F}_z$ for $ z \in \F_2^{m-r+i}$ be the event that $\widehat{Y}_i^{(j)}(z) \neq 0$. Now, via arguments similar to those in \cite[Sec. IV-B]{rpa-analysis-isit}, it can be checked that conditioned on the event $\mathcal{G}_i$, the ``aggregated" LLR vector $\widehat{L}_i^{(j)}$ at node $(i,j)$ obeys
\begin{equation}
    \widehat{L}_i^{(j)}(z) = \frac{1}{2^{m-r+i }-1} \sum_{z'\neq z } L_i^{(j)}(z'). \notag
\end{equation}
Let $\overline{L}_i^{(j)}(z):= (2^{m-r+i }-1)\cdot \widehat{L}_i^{(j)}(z)$. We then have via the Markov inequality that
    \begin{align}
         \PR{\mathcal{F}_z| \cG_i} & =  \PR{ \overline{L}_i^{(j)}(z)  <0 }  \notag \\ 
        & \leq \E{e^{-\overline{L}_i^{(j)}(z)/2} } \notag\\
        &= \prod_{z\neq z'}  \E{e^{L_i^{(j)}(z')/2}} = Z( W_i ) ^{ 2^{m-r+i} -1 }. \notag
    \end{align}

Now, conditioned on $\cG_i$, note that we have $Q_i = \cup_{z\in \F_2^{m-r+i}} \mathcal{F}_z$. Hence, 
\begin{align*}
\PR{ { \cQ_i } | \cG_i} &\leq  \sum_{z\in \F_2^{m-r+i}} \PR{ \mathcal{F}_z | \cG_i}\\ &= 2^{m-r+i } Z( W_i ) ^{ 2^{m-r+i} -1 },
\end{align*}
thereby proving our second claim. 
\end{proof}
\subsection{Explicit Upper Bound on $\PR{ \cQ_i }$}
The following lemma then follows by ``unrolling" the recursion in Lemma \ref{lem:Qi-rec}.
\begin{lemma}\label{lem:Qi-unroll}
Define for each $t\in\{2,\dots,r\}$,
\[
A_t := 2^{m-r+t}\, Z_t^{\,2^{m-r+t}-1},
\qquad
B_t := 2^{m-r+t}-1.
\]
Then for every $i\in\{2,\dots,r\}$,
\begin{equation*}\label{eq:Qi-unrolled}
\PR{\cQ_i}
\le
\sum_{t=2}^{i}\Bigg( A_t \prod_{s=t+1}^{i} B_s \Bigg)
+
\PR{\cQ_1}\prod_{s=2}^{i} B_s .
\end{equation*}
In particular, 
\begin{equation*}\label{eq:Qr-unrolled}
\PR{\cQ_r}
\le
\sum_{t=2}^{r}\Bigg( A_t \prod_{s=t+1}^{r} B_s \Bigg)
+
\PR{\cQ_1}\prod_{s=2}^{r} B_s .
\end{equation*}
\end{lemma}

\section{Proof of Main Result}

In this section we will prove Thm.~\ref{thm:main}. 

\begin{proof}[Proof of Thm.~\ref{thm:main}]
    For each $ t \in  \{ 2 , \ldots, r \}$ from Lemma~\ref{lem:bhat-unroll} we have that $A_t := 2^{m-r+t}  Z_t^{2^{m-r+t}-1}$ and $B_t := 2^{m-r+t}-1 \leq 2^{m-r+t}.$
    Observe that for any $u\in \{2,\ldots,r\}$, we have
    \begin{align*}
       \prod_{s=u}^r B_s &=  \prod_{s=u}^r 2^{m-r+s} \\&=2^{(r-u+1)(m-r)+\sum_{s=u}^r s}
\leq 2^{O(mr)}.
    \end{align*}
Thus, from Lemma \ref{lem:Qi-unroll},
\begin{equation}
\label{eq:temp}
    \PR{\cQ_r}
\leq 2^{O(mr)} \PR{ \cQ_1} + \sum_{t=2 }^r 2^{ O(mr)} Z_t^{2^{m-r+t}-1}.
\end{equation}

We will individually show that each term above approaches zero in the large $m$ limit, for $r \le \log_2 m -\log_2 \lambda-\delta$, for some small constant $\delta$. For the first term we will show that $\PR{\cQ_1}$ is very small compared to its multiplicative pre-factor that is $2^{ O(mr)}$. To this end, recall that $ \lambda := -\ln( 1- Z) = -\ln( 1- Z_r)$. From Lemma~\ref{lem:bhat-unroll}, we know that
$
    ( 1 - Z_t ) \geq ( 1 - Z_r)^{2^{r-t}} = e^{-\lambda \cdot 2^{r-t}}.
$
Now, from Lemma~\ref{lem:Q1}, we obtain using the previous inequality that
\begin{equation*}\label{eq:asym-Q1}
\begin{aligned}
        \PR{ \cQ_1} 
       & \leq  2^{m-r+2} e^{-(2^{m-r}\cdot  e^{  -\lambda\cdot 2^{r-1}})}.
       \end{aligned}
\end{equation*} 
Here, we use the fact that for any positive integer $k$, $ Z_t^{k} \leq e^{ -k ( 1- Z_t)  } $. Notice that in the regime $r \le \log_2 m -\log_2 \lambda -\delta$, we have $ \lambda\cdot 2^{r-1} = O(m)$. Therefore,
$
      2^{m-r}\cdot\exp(-\lambda 2^{r-1})=\exp(\Omega(m)).
$
Plugging this into the first term in \eqref{eq:temp}, we get that
\begin{equation*}
    \begin{aligned}
        2^{O(mr)} \PR{ \cQ_1}  &\leq 2^{O(mr)}  2^{m-r+2} e^{-(2^{m-r}  e^{  -\lambda 2^{r-1}})}\\
        & \leq 2^{O(mr)} e^{-e^{\Omega (m) }}\xrightarrow{m\to \infty} 0.
    \end{aligned}
\end{equation*}

Now we focus on the second summation in \eqref{eq:temp}; we will show that each summand approaches zero. So fix any $t$, define $b_t = 2^{ m-r + t } - 1 $. Then, since
$
    Z_t^{b_t} \leq e^{- b_t(1-Z_t) },
$
we have that
\begin{equation*}
\begin{aligned}  
        (1-Z_t ) b_t  &\geq \exp (-\lambda 2^{r-t} ) (  2^{m-r+t}  - 1 ) \\
        &\geq   2^{m-r+t-1 }\cdot \exp (-\lambda\cdot 2^{r-t} )\\
        &= \exp\left({ ( m - r + t -1 )\ln 2 - \lambda\cdot 2^{r-t}}\right).
 \end{aligned}
\end{equation*}
Again using the fact that $ \lambda\cdot 2^{r-t}\leq \lambda\cdot 2^{r-1} = o(m) $, we get that
\begin{equation*}
    ( m - r + t -1 )\ln 2 - \lambda\cdot 2^{r-t} = \Theta(m) - o(m) = \Omega(m). 
\end{equation*}
Plugging this into any term in the summation in \eqref{eq:temp}, we get that 
\begin{equation*}
    \begin{aligned}
        2^{ O(mr) } Z_t^{2^ { m-r+t} -1 }  &\leq 2^{O(mr) } e^{ -\exp( m - r + t -1 )\ln 2 - \lambda 2^{r-t})} \\ 
        &\xrightarrow{m\to \infty} 0.
    \end{aligned}
\end{equation*}
This concludes the proof of Theorem \ref{thm:main}.
\end{proof}

\section*{Acknowledgments}
V.~A.~Rameshwar acknowledges support from the New Faculty Initiation Grant, IIT Madras.  D.~Fathollahi acknowledges support from the National Science Foundation (NSF) under Grant No.~CCF-2231157.
The authors thank Harshithanjani Athi for helpful discussions on the problem.

\bibliographystyle{IEEEtran}
{\footnotesize
	\bibliography{references}}
\end{document}

%% file: Macros.tex
\usepackage[margin=1.125in]{geometry}
\usepackage[table]{xcolor}
\definecolor{DarkGreen}{rgb}{0.1,0.5,0.1}
\definecolor{DarkRed}{rgb}{0.5,0.1,0.1}
\definecolor{DarkBlue}{rgb}{0.1,0.1,0.5}

\usepackage[small]{caption}
\usepackage[pdftex]{hyperref}
\hypersetup{
    unicode=false,          
    pdftoolbar=true,        
    pdfmenubar=true,        
    pdffitwindow=false,      
    pdfnewwindow=true,      
    colorlinks=true,       
    linkcolor=DarkBlue,          
    citecolor=DarkGreen,        
    filecolor=DarkGreen,      
    urlcolor=DarkBlue,          
    pdftitle={},
    pdfauthor={},    
}
\usepackage{amssymb,amsmath,amsthm,amsfonts,enumitem}
\usepackage{fullpage,nicefrac,comment}
\usepackage{tikz}
\usetikzlibrary{arrows,decorations.pathmorphing,decorations.shapes,snakes,patterns,positioning, shapes.callouts, shapes.arrows}
\usepackage[noend]{algpseudocode}
\usepackage{algorithm}
\usepackage{bbm}


\newcommand{\bits}{\left\{0,1\right\}}

\newcommand{\cC}{\ensuremath{\mathcal{C}}}
\newcommand{\cE}{\ensuremath{\mathcal{E}}}
\newcommand{\cG}{\ensuremath{\mathcal{G}}}

\newcommand{\cQ}{\ensuremath{\mathcal{Q}}}

\newcommand{\F}{{\mathbb F}}

\newcommand{\cY}{\mathcal{Y}}

\newcommand{\PR}[1]{{\mathbb{P}}\left\{ #1\right\}}
\newcommand{\E}[1]{\mathbb{E}\left[ #1\right]}

\newcommand{\argmax}{\mathrm{argmax}}

\renewcommand{\epsilon}{\varepsilon}

\newtheorem{theorem}{Theorem} 
\newtheorem{lemma}[theorem]{Lemma} 
\newtheorem*{lemma*}{Lemma}
\newtheorem{definition}{Definition}

\newtheorem{remark}{Remark}

\newcommand{\indic}[1]{\mathbbm{1} [ #1] }